\documentclass[onecolumn,preprintnumbers,amsmath,amssymb]{revtex4}
\usepackage{array}
\usepackage{booktabs}
\usepackage{tabu}
\usepackage{dcolumn}
\usepackage{amsmath}
\usepackage{amsfonts}
\usepackage{amssymb}
\usepackage{graphicx}
\usepackage{subfigure}
\usepackage{graphicx}% Include figure files
\usepackage{dcolumn}% Align table columns on decimal point
\usepackage{bm}% bold math
\usepackage{xcolor}
\usepackage{amsthm}
\def\be{\begin{equation}}
\def\ee{\end{equation}}
\def\bea{\begin{eqnarray}}
\def\eea{\end{eqnarray}}
\def\f{\frac}
\def\n{\nonumber}
\def\l{\label}
\def\p{\phi}
\def\o{\over}
\def\R{\rho}
\def\pa{\partial}
\def\om{\omega}
\def\na{\nabla}
\def\P{\Phi}
\newtheorem{theorem}{Theorem}
\begin{document}

\title{Lower and upper bounds for unilateral coherence and applying them to the  entropic uncertainty relations}% Force line breaks with \\

\author{H. Dolatkhah}
\affiliation{Department of Physics, University of Kurdistan, P.O.Box 66177-15175, Sanandaj, Iran}
\author{S. Haseli}
\affiliation{Faculty of Physics, Urmia University of Technology, Urmia, Iran}
\author{S. Salimi}
\email{shsalimi@uok.ac.ir}
\author{A. S. Khorashad}
\affiliation{
Department of Physics, University of Kurdistan, P.O.Box 66177-15175, Sanandaj, Iran.\\}
\date{\today}% It is always \today, today,

\def\be{\begin{equation}}
  \def\ee{\end{equation}}
\def\bea{\begin{eqnarray}}
\def\eea{\end{eqnarray}}
\def\f{\frac}
\def\n{\nonumber}
\def\l{\label}
\def\p{\phi}
\def\o{\over}
\def\R{\rho}
\def\pa{\partial}
\def\om{\omega}
\def\na{\nabla}
\def\P{\Phi}
%\nofiles

%=============================================================%
%=============================================================%
%============== Abstract =======================================%
%=============================================================%
%=============================================================%
\begin{abstract}
The uncertainty principle sets a bound on our ability to predict the measurement outcomes of two incompatible observables which are measured on a quantum particle simultaneously. In quantum information theory, the uncertainty principle can be formulated in terms of the Shannon entropy. Entropic uncertainty bound can be improved by adding a particle which correlates with the measured particle. The added particle acts as a quantum memory. In this work, a method is provided for obtaining the entropic uncertainty relations in the presence of a quantum memory by using quantum coherence. In the method, firstly, one can use the quantum relative entropy of quantum coherence to obtain the uncertainty relations. Secondly, these relations are applied to obtain the entropic uncertainty relations in the presence of a quantum memory. In comparison with other methods this approach is much simpler. Also, for a given state, the upper bounds on the sum of the relative entropies of unilateral coherences are provided, and it is shown which one is tighter. In addition, using the upper bound obtained for unilateral coherence, the nontrivial upper bound on the sum of the entropies for different observables is derived in the presence of a quantum memory.

\end{abstract}
%\pacs{04.50}
%\keywords{keyword.}%Use showkeys class option if keyword
\maketitle

%%%%%%%%%%%%%%%%%%%%%%%%%%%%%%%%%%%%%%%%%%%%%%%%%%%%%%%%%%%%%%%%%%%%%%%%%%%%
%%%%%%%%%%%%%%%%%%%%%%%%%%%%%%%%%%%%%%%%%%%%%%%%%%%%%%%%%%%%%%%%%%%%%%%%%%%%
%%%%%%%%%%%%%%%%%%%%%%%%%%%%%%%%%%%%%%%%%%%%%%%%%%%%%%%%%%%%%%%%%%%%%%%%%%%%
%%%%%%%%%%%%%%%%%%%%%%%%%%%%%%%%%%%%%%%%%%%%%%%%%%%%%%%%%%%%%%%%%%%%%%%%%%%%
%============  Sec.I (Introduction)  =======================================
%%%%%%%%%%%%%%%%%%%%%%%%%%%%%%%%%%%%%%%%%%%%%%%%%%%%%%%%%%%%%%%%%%%%%%%%%%%%
%%%%%%%%%%%%%%%%%%%%%%%%%%%%%%%%%%%%%%%%%%%%%%%%%%%%%%%%%%%%%%%%%%%%%%%%%%%%
%%%%%%%%%%%%%%%%%%%%%%%%%%%%%%%%%%%%%%%%%%%%%%%%%%%%%%%%%%%%%%%%%%%%%%%%%%%%
%%%%%%%%%%%%%%%%%%%%%%%%%%%%%%%%%%%%%%%%%%%%%%%%%%%%%%%%%%%%%%%%%%%%%%%%%%%%
\section{INTRODUCTION}	%) A SECTION HEADING
The uncertainty  principle is one of the most important topics in quantum theory \cite{Heisenberg}. According to this principle, our ability to predict the measurement outcomes of two incompatible observables, which simultaneously are measured on a quantum system, is restricted. This statement for the uncertainty principle is formulated by the uncertainty relation. Robertson and Schr\"{o}dinger have shown that for any two incompatible observables, $X$ and $Z$, the uncertainty relation has the following form \cite{Robertson,Schrodinger}, 
\begin{equation}\label{Roberteq}
\Delta X \Delta Z \geqslant \frac{1}{2}\vert \langle \left[ X, Z \right] \rangle \vert,
\end{equation}
where $\Delta X, \Delta Z$ are the standard deviations of $X$ and $Z$, 
\begin{equation}\label{deviation}
\Delta X = \sqrt{\langle X^{2}\rangle -\langle X \rangle^{2}} \quad \rm{and} \quad \Delta Z = \sqrt{\langle Z^{2}\rangle -\langle Z \rangle^{2}},
\end{equation}
respectively. Due to the state-dependent characteristic of this lower bound, the bound can vanish even when $X$ and $Z$ are incompatible. In order to overcome this roughness, Deutsch applied the Shannon entropy (instead of the standard deviation) to derive the uncertainty relation \cite{Deutsch}. Deutsch introduced the entropic uncertainty relation (EUR) which is the counterpart of the Heisenberg uncertainty relation in quantum information theory. Deutsch's inequality was improved by Massen and Uffink as \cite{Uffink}
\begin{equation}\label{Maassen and Uffink}
H(X)+H(Z)\geqslant \log_2 \frac{1}{c}\equiv q_{MU},
\end{equation}
where $H(O) = -\sum_{k} p_k \log_2 p_k$ is the Shannon entropy of the measured observable $O \in \lbrace X, Z \rbrace$, $p_k$ is the probability of the outcome $k$, the quantity $c$ is defined as $c = \max_{\lbrace \mathbb{X},\mathbb{Z}\rbrace } \vert\langle x \vert z\rangle \vert ^{2}$, in which $\mathbb{X}=\lbrace \vert x\rangle \rbrace$ and $\mathbb{Z}=\lbrace \vert z\rangle \rbrace$ are the eigenbases of $X$ and $Z$, respectively, and $q_{MU}$ is called incompatibility measure. \\
Various efforts have been done to improve this relation \cite{Berta,Coles1,Bialynicki,Pati,Ballester,Vi,Wu,Wehner,Rudnicki,Rudnicki1,Pramanik,Maccone,Pramanik1, Zozor,Coles,Adabi,Adabi1,Yunlong,Liu,Kamil,Zhang,R}. Berta \emph{et al.} studied EUR in the presence of a quantum memory \cite{Berta}. They considered an extra quantum system as the quantum memory $B$ correlating with the measured quantum system $A$. The strategy can be explained in terms of the uncertainty game between two players, Alice and Bob. In the game, Bob prepares a correlated bipartite state $\rho_{AB}$ and sends the part $A$ to Alice and keeps the other one as the quantum memory. Alice performs a measurement on her part by choosing one of $X$ and $Z$, then she informs Bob on her choice of the measurement. Bob's task is to minimize his uncertainty about the outcome of Alice's measurement . The EUR in the presence of a quantum memory can be written as \cite{Berta}  
\begin{equation}\label{Berta}
H(X \vert B)+H(Z \vert B) \geqslant q_{MU} +S(A \vert B),
\end{equation}
where $H(O \vert B) = S(\rho_{OB})-S(\rho_{B})$ $(O \in \lbrace X, Z \rbrace)$ and $S(A|B) = S(\rho_{AB})-S(\rho_{B})$ are the conditional von Neumann entropies. The post measurement states after measuring $X$ or $Z$ on the part $A$ are given by 
\begin{equation}
\rho_{XB}= \sum_{x}(\vert x\rangle\langle x\vert_{A}\otimes \mathbf{I}_B ) \rho_{AB}(\vert x\rangle\langle x\vert_{A}\otimes \mathbf{I}_{B} ),\nonumber
\end{equation}
or 
\begin{equation}
\rho_{ZB}= \sum_{z}(\vert z\rangle\langle z\vert_{A}\otimes \mathbf{I}_{B} ) \rho_{AB}(\vert z\rangle\langle z\vert_{A}\otimes \mathbf{I}_{B} ),\nonumber
\end{equation}
respectively, where $\mathbf{I}_{B}$ is the identity operator. Note that when the conditional entropy $S(A \vert B)$ is negative which means that the part $A$ and the quantum memory $B$ are entangled, the uncertainty lower bound reduces, therefore, Bob can predict Alice's measurement outcomes with better accuracy. Moreover, when particle $A$ is maximally entangled with the memory $B$, Bob can accurately predict Alice's measurement outcomes. Also, if there is no quantum memory, Eq.\;(\ref{Berta}) reduces to
\begin{equation}\label{Berta2}
H(X)+H(Z)\geqslant q_{MU}+S(A),
\end{equation}
which is tighter than the Maassen and Uffink EUR.

Much effort has been made to improve the lower bound of the EUR in the presence of a quantum memory \cite{Pati,Adabi,Coles}. Pati \emph{et al.} improved the EUR in the presence of a quantum memory by adding a term to the lower bound in Eq.\;(\ref{Berta}) \cite{Pati},
\begin{equation}\label{pati}
H(X \vert B)+H(Z \vert B) \geqslant q_{MU} + S(A \vert B)+ \max\lbrace 0,D_{A}(\rho_{AB})-J_{A}(\rho_{AB}) \rbrace,
\end{equation}
where $D_{A}(\rho_{AB}):=I(A:B)-J_{A}(\rho_{AB})$ is quantum discord, $I(A:B)=S(\rho_{A})+S(\rho_{B})-S(\rho_{AB})$  is mutual information, and
\begin{equation}\label{ss}
J_{A}(\rho_{AB})=S(\rho_{B})-\min_{\lbrace \Pi_{i}^{A} \rbrace}S(\rho_{B\vert \lbrace \Pi_{i}^{A} \rbrace})
\end{equation} 
is the classical correlation between $A$ and $B$. The minimization is taken over the set of all possible positive operator-valued measures (POVMs) on the subsystem $A$. In Eq.(\ref{ss}), $S(\rho_{B\vert \lbrace \Pi_{i}^{A} \rbrace})=\sum_{i}p_{i}S(\rho_{B|i})$, where the post measurement state  $\rho_{B|i}=  \frac{Tr_{A} \big( (\vert i\rangle\langle i \vert_{A}\otimes \mathbf{I}_{B})\rho_{AB}(\vert i\rangle\langle i \vert_{A}\otimes \mathbf{I}_{B}) \big)}{p_{i}}$  is obtained with the probability $p_{i}= Tr_{AB}\big( (\vert i\rangle\langle i \vert_{A}\otimes \mathbf{I}_{B})\rho_{AB}(\vert i\rangle\langle i \vert_{A}\otimes \mathbf{I}_{B})\big)$.  Note that if $D_{A}(\rho_{AB})>J_{A}(\rho_{AB})$ Pati's lower bound is tighter than that of Berta.

Adabi \emph{et al.} provided a lower bound for the EUR in the presence of a quantum memory by considering an additional term on the right-hand side of Eq.\;(\ref{Berta}) \cite{Adabi},
\begin{equation}\label{new1}
H(X \vert B)+H(Z \vert B)\geqslant q_{MU} + S(A|B)+\max\{0 , \delta\},
\end{equation}
where  $$\delta=I(A:B)-[I(X:B)+I(Z:B)],$$
and
$$I\big(O:B\big)= S(\rho_{B})- \sum_{o}p_{o}S(\rho_{B|o}), O \in \lbrace X, Z \rbrace$$
 is the Holevo quantity. It is actually equal to the upper bound of the accessible information to Bob about Alice's measurement outcomes. They showed that for many states such as the Bell diagonal and the Werner states, this lower bound is tighter than those of Berta and Pati. \\
The EUR has many applications in the field of quantum information, including quantum key distribution \cite{Koashi,Berta}, quantum cryptography \cite{Dupuis,Koenig}, quantum randomness \cite{Vallone,Cao}, entanglement witness \cite{Berta2}, EPR steering \cite{Walborn,Schneeloch}, and quantum metrology \cite{Giovannetti}.\\ 

The motivation of this work is to provide a method which can be applied to obtain the EURs, in the presence of a quantum memory, from the  uncertainty relations for the relative entropy of quantum coherence. Also, for a given bipartite quantum state, the upper bounds on the sum of the relative entropies of unilateral coherences are provided which can be used to derive a nontrivial upper bound on the sum of the entropies for two incompatible observables in the presence of a quantum memory. The paper is organized as follows: In Sec. \ref{Sec2} a brief review on quantum coherence  and purity from the viewpoint of the resource theory is presented. In Sec. \ref{Sec3} several lower and upper bounds on the sum of the relative entropies of unilateral coherences of a given state with respect to two different measurement bases are derived. In Sec. \ref{Sec4}, it is shown that how one can obtain the  EURs in the presence of a quantum memory by using  the uncertainty relations of a quantum coherence. Finally, the results are summarized in Sec. \ref{conclusion}.

\section{Preliminary}\label{Sec2}
\subsection{Resource theory of quantum coherence}
All quantum resource theories have two main ingredients: (1) free states which are unable to provide the resources for quantum information processing and (2) free operations that map the set of free states onto itself. Quantum coherence is undoubtedly a cornerstone of quantum physics which differentiates quantum mechanics from the classical world. In this section, we review the resource framework for quantum coherence \cite{Baumgratz,Streltsov}. Let us consider a $d$-dimensional Hilbert space $\mathcal{H}_{d}$ and an arbitrary fixed reference basis $\mathbb{Y}=\lbrace \vert y \rangle \rbrace_{y=1,...,d}$. The free states which are called  incoherent states, can be written as
 \begin{equation}
 \delta =\sum_{y=1}^{d} \delta_{y}\vert y \rangle\langle y \vert, 
 \end{equation}
where $\delta_{y} \in [0,1] $ and $\sum_{y=1}^{d} \delta_{y}=1$.  The  set of incoherent states is denoted by $\mathcal{I}$. The states which are not of the above form are called coherent states. In the resource theory of quantum coherence, the definition of free operations is not unique, and in fact several definitions have been provided in the literature \cite{Streltsov}. The most important free operations of this resource theory are called incoherent operations introduced by Baumgratz \emph{et al.} \cite{Baumgratz}, and have the Kraus representation form as $ \Lambda_{IO}[\rho]=\sum_{m} K_{m} \rho K_{m}^{\dag}$, where $\Lambda_{IO}$ is an  incoherent operation and $K_{m} \mathcal{I} K_{m}^{\dag} \subseteq \mathcal{I}$.\\ A proper measure of coherence, $C(\rho)$, should satisfy the following conditions \cite{Baumgratz,Streltsov}
\begin{enumerate}
\item $C(\rho)\geqslant 0$ for any state $\rho$ and $C(\rho) = 0$ iff $\rho$ is the incoherent state.
\item $C(\rho)$ is nonincreasing under the incoherent operations,
i.e., $C(\rho) \geqslant C( \Lambda_{IO} [\rho] )$.
\item $C(\rho)$ is nonincreasing on average under the selective incoherent operations, i.e., $ C(\rho) \geqslant \sum_{m} p_{m}C(\rho_{m})$ in which $p_{m}=Tr[K_{m} \rho K_{m}^{\dag}]$, $ \rho_{m}=K_{m} \rho K_{m}^{\dag} / p_{m}$ and $K_{m} \mathcal{I} K_{m}^{\dag} \subseteq \mathcal{I}$.
\item $C(\rho)$ is convex, i.e., $\sum_{i} p_{i}C(\rho_{i})\geqslant C(\sum_{i} p_{i}\rho_{i})$ 
\end{enumerate}
Over the last few years different measures have been proposed for quantum coherence \cite{Baumgratz,Streltsov,Yao,Rana,Napoli}; one of the most important measures is the relative entropy of coherence \cite{Baumgratz} which is defined as
 \begin{equation}
C_{rel}(\mathbb{Y}\vert\rho)=\min_{\delta \in \mathcal{I}} S(\rho\Vert \delta),
\end{equation}
where $S(\rho\Vert \delta)=Tr[\rho\log_{2}\rho]-Tr[\rho\log_{2}\delta]$ is the relative entropy. The relative entropy of coherence can be expressed in a simple form as \cite{Baumgratz}
\begin{equation}\label{rc} 
C_{rel}(\mathbb{Y}\vert\rho)=S\big(\rho\Vert \Delta_{\mathbb{Y}}(\rho)\big)=S\big(\Delta_{\mathbb{Y}}(\rho)\big)-S(\rho),
\end{equation}
where $\Delta_{\mathbb{Y}}(\rho)=\sum_{y} \langle y \vert\rho\vert y \rangle \vert y \rangle\langle y \vert$ is obtained from removing all off-diagonal elements of the state density matrix. Note that $S\big(\Delta_{\mathbb{Y}}(\rho)\big)$ is the Shannon entropy, $H(Y)$, with the probabilities $p_{y}=\langle y \vert\rho\vert y \rangle$. Thus Eq. (\ref{rc}) can be rewritten as
\begin{equation}\label{relative coherence}
C_{rel}(\mathbb{Y}\vert\rho)=H(Y)-S(\rho).
\end{equation}

One can extend the concept of coherence to bipartite system \cite{Streltsov2,Bromley,Chitambar,Teng Ma,Jiajun Ma}. For a bipartite system $AB$, it is possible to consider the coherence with respect to a local basis of one of the subsystems, for example $A$. This is called unilateral coherence \cite{Teng Ma}. Thus in the bipartite system $AB$, incoherent-quantum state with respect to the basis of the subsystem $A$ ($\lbrace\vert y \rangle _{A}\rbrace$ ) is defined as \cite{Jiajun Ma,Chitambar}
\begin{equation}
\sigma_{AB}\in \mathcal{I}_{B|A}, \quad \sigma_{AB}=\sum_{y}p_y \vert y\rangle\langle y \vert_{A}\otimes \sigma_{B|y},
\end{equation}
where $\mathcal{I}_{B|A}$ is the set of incoherent-quantum states, and $\sigma_{B|y}$ is an arbitrary state of the subsystem $B$. The relative entropy of unilateral coherence, $C_{rel}^{B|A}$, is a measure of unilateral coherence for a bipartite density matrix $\rho_{AB}$ and is given by
\begin{equation}\label{reou}
C_{rel}^{B|A}(\mathbb{Y}\vert\rho_{AB})=\min_{\sigma_{AB}\in \mathcal{I}_{B|A} }S(\rho_{AB}||\sigma_{AB}).
\end{equation}
Minimizing the right-hand side of Eq. (\ref{reou}), leads us to \cite{Chitambar}
\begin{equation}\label{reou2}
C_{rel}^{B|A}(\mathbb{Y}\vert\rho_{AB})=S\big(\rho_{AB}||\Delta_{\mathbb{Y}}(\rho_{AB})\big)=S(\rho_{YB})-S(\rho_{AB}),
\end{equation}
where 
\begin{equation}\nonumber
\Delta_{\mathbb{Y}}(\rho_{AB})=\rho_{YB}=\sum_{y}\left( \vert y\rangle\langle y \vert_{A}\otimes \mathbf{I}_{B})\rho_{AB}(\vert y\rangle\langle y \vert_{A}\otimes \mathbf{I}_{B}\right) ,
\end{equation}
or
\begin{equation}\nonumber
\rho_{YB}=\sum_{y}p_y \vert y\rangle\langle y \vert_{A}\otimes \rho_{B|y},
\end{equation}
in which $p_{y}= Tr_{AB} \big( (\vert y\rangle\langle y \vert_{A}\otimes \mathbf{I}_{B})\rho_{AB}(\vert y\rangle\langle y \vert_{A}\otimes \mathbf{I}_{B}) \big),$ and $\rho_{B|y}= \frac{Tr_{A} \big( (\vert y\rangle\langle y \vert_{A}\otimes \mathbf{I}_{B})\rho_{AB}(\vert y\rangle\langle y \vert_{A}\otimes \mathbf{I}_{B}) \big)}{p_{y}}$.\\
Regarding Eqs. (\ref{relative coherence}) and (\ref{reou2}), one can obtain  
\begin{equation}\label{unilateral coherence}
C_{rel}^{B|A}(\mathbb{Y}\vert\rho_{AB})=C_{rel}(\mathbb{Y}\vert\rho_{A})+I(A:B)-I(Y:B),
\end{equation} 
which indicates the relation between unilateral coherence and the coherence of a subsystem $(A)$.
\subsection{Resource theory of purity}
Now, let us briefly review the concept of purity from the resource theory point of view. In this resource theory the free state is the maximally mixed state, $\dfrac{\mathbf{I}}{d}$ and the free operations can be unital operations, mixture of unitary operations and noisy operations \cite{Gour,Streltsov3}. Several measures of purity have been introduced; the most important example is the relative entropy of purity defined as
\begin{equation}\label{purity}
\mathcal{P}_{rel}(\rho)=S(\rho\Vert \dfrac{\mathbf{I}}{d})=\log_{2} d-S(\rho).
\end{equation}
As can be seen, there is no minimization in Eq.\;(\ref{purity}) due to the uniqueness of the free state in this resource theory.
Since $\dfrac{\mathbf{I}}{d}\in\mathcal{I}$, the relative entropy of purity is an upper bound of the relative entropy of coherence,  
\begin{equation}\label{puritycoherence}
\mathcal{P}_{rel}(\rho)\geqslant C_{rel}(\mathbb{Y}\vert\rho).
\end{equation}
For a bipartite system $AB$, one can define unilateral purity with respect to the subsystem $A$. In this case, the free state is $\dfrac{\mathbf{I}_{A}}{d}\otimes \rho_{B},$ where $d$ is the dimension of the subsystem $A$ and $\rho_{B}$ is an arbitrary state of the subsystem $B$. The relative entropy of unilateral purity is then defined as
\begin{equation}\label{uni purity}
\mathcal{P}_{rel}^{B|A}(\rho_{AB})=S(\rho_{AB}\Vert \dfrac{\mathbf{I}_{A}}{d}\otimes \rho_{B} )=\log_{2} d-S(A \vert B).
\end{equation}
From Eqs.\;(\ref{purity}) and (\ref{uni purity}), one can obtain the following expression for the unilateral purity,
\begin{equation}
\mathcal{P}_{rel}^{B|A}(\rho_{AB})=\mathcal{P}_{rel}(\rho_{A})+I(A:B).
\end{equation}
This means that the unilateral purity is equal to the sum of the purity of the subsystem $A$ and the mutual information.
%%%%%%%%%%%%%%%%%%%%%%%%%%
%%%%%%%%%%%%%%%%%%%%%%%%%%%
%%%%%%%%%%%%%%%%%%%%%%%%%%
%%%%%%%%%%%%%%%%%%%%%%%%%%
\section{Entropic uncertainty relations for Quantum coherence}\label{Sec3}
Many of the uncertainty relations proposed for the relative entropy of coherence are based on the EURs \cite{Singh,Cheng,Dolatkhah,Yuan}. In this section, regarding the properties of  quantum coherence, uncertainty relations for the relative entropy of coherence of two  incompatible observables (two incompatible quantum measurements bases) are obtained. The nonincreasing behavior of the relative entropy under quantum operations is the most important property applied to obtain uncertainty relations for the relative entropy of coherence in the following theorems.\\ 
\begin{theorem}
Given two general measurements $\mathbb{X}$ and $\mathbb{Z}$, the following uncertainty relation for the relative entropy of coherence holds for any state $\rho$,
\begin{equation}\label{unrc1}
C_{rel}(\mathbb{X}\vert\rho)+C_{rel}(\mathbb{Z}\vert\rho)\geqslant q_{MU} -S(\rho).
\end{equation}
\end{theorem}
\begin{proof}
Approach is similar to the one used in \cite{K. Korzekwa}
\begin{eqnarray}\small
C_{rel}(\mathbb{X}\vert\rho)+C_{rel}(\mathbb{Z}\vert\rho)&=&S\big(\rho\Vert\Delta_{\mathbb{X}}(\rho)\big)+S\big(\rho\Vert\Delta_{\mathbb{Z}}(\rho)\big) \nonumber \\
 & \geqslant & S\big(\rho\Vert\Delta_{\mathbb{X}}(\rho)\big)+S\big(\Delta_{\mathbb{X}}(\rho)\Vert\Delta_{\mathbb{X}}(\Delta_{\mathbb{Z}}(\rho))\big) \nonumber \\
 &=&S\big(\rho\Vert\Delta_{\mathbb{X}}(\rho)\big)+S\left( \Delta_{\mathbb{X}}(\rho)\Vert\sum_{x,z}\vert\langle x \vert z\rangle \vert ^{2}\langle z \vert\rho\vert z \rangle \vert x \rangle\langle x \vert\right) \nonumber \\
 &\geqslant & S\big(\rho\Vert\Delta_{\mathbb{X}}(\rho)\big)+S\big(\Delta_{\mathbb{X}}(\rho)\Vert c \mathbf{I}\big)=-S(\rho)+\log_2 \frac{1}{c},
\end{eqnarray}
where the first and the second inequalities follow from the fact that the relative entropy is both contractive under completely positive maps \cite{Nielsen} and nonincreasing function under of its second argument (see lemma 5 of Ref.\;\cite{Coles33}), respectively.
\end{proof} 
Considering $H(X)=C_{rel}(\mathbb{X}\vert\rho)+S(\rho)$ along with Eq.\;(\ref{unrc1}) leads us to the following EUR
\begin{equation}\label{Berta22222}
H(X)+H(Z)\geqslant q_{MU}+S(A).
\end{equation}
As can be seen, the above derivation of Eq.\;(\ref{Berta22222}) is much easier than what has been provided in \cite{Berta}. Similarly, uncertainty relation can be obtained for unilateral coherence.\\ 
\begin{theorem}

Let $X$ and $Z$ be two incompatible observables with bases $\mathbb{X}$ and $\mathbb{Z}$, respectively. The following uncertainty relation for the relative entropy of unilateral coherence of these observables holds for any state $\rho_{AB}$,

\begin{equation}\label{berta coherence}
C_{rel}^{B|A}(\mathbb{X}\vert\rho_{AB})+C_{rel}^{B|A}(\mathbb{Z}\vert\rho_{AB})\geqslant q_{MU} -S(A\vert B),
\end{equation}
\end{theorem} 
\begin{proof}
The proof is similar to the proof of the previous theorem.
\begin{eqnarray}
C_{rel}^{B|A}(\mathbb{X}\vert\rho_{AB})+C_{rel}^{B|A}(\mathbb{Z}\vert\rho_{AB})&=&S\left( \rho_{AB}\Vert\Delta_{\mathbb{X}}(\rho_{AB})\right) +S\left( \rho_{AB}\Vert\Delta_{\mathbb{Z}}(\rho_{AB})\right) \nonumber \\
&\geqslant & S\left( \rho_{AB}\Vert\Delta_{\mathbb{X}}(\rho_{AB})\right) +S\left( \Delta_{\mathbb{X}}(\rho_{AB})\Vert\Delta_{\mathbb{X}}(\Delta_{\mathbb{Z}}(\rho_{AB}))\right)  \nonumber \\
&=& S\left( \rho_{AB}\Vert\Delta_{\mathbb{X}}(\rho_{AB})\right) +S\left( \Delta_{\mathbb{X}}(\rho_{AB})\Vert \sum_{x,z}p_z \vert\langle x \vert z\rangle \vert ^{2}\vert x\rangle\langle x \vert_{A}\otimes \rho_{B|z}\right)  \nonumber \\
 & \geqslant & S\left( \rho_{AB}||\rho_{XB}\right) +S\left( \rho_{XB}||c \mathbf{I}_{A}\otimes \rho_{B}\right)  \nonumber \\
& =& -S(A\vert B)+\log_2 \frac{1}{c}.
 \end{eqnarray}

\end{proof}
Using Eqs.\;(\ref{unilateral coherence}) and (\ref{unrc1}), one can propose a different uncertainty relation for the relative entropy of unilateral coherence, whose lower bound is tighter than that of Eq.\;(\ref{berta coherence}).\\ 
\begin{theorem}
Given two noncommuting observables $X$ and $Z$ with eigenbases $\mathbb{X}$ and $\mathbb{Z}$, respectively. The following uncertainty relation for the relative entropy of unilateral coherence of these observables holds for any state $\rho_{AB}$,
\begin{equation}\label{pati coherence}
C_{rel}^{B|A}(\mathbb{X}\vert\rho_{AB})+C_{rel}^{B|A}(\mathbb{Z}\vert\rho_{AB}) \geqslant  q_{MU} - S(A \vert B)+\max\lbrace 0,D_{A}(\rho_{AB})-J_{A}(\rho_{AB}) \rbrace.
\end{equation}
\end{theorem}
\begin{proof}According to Eq.\;(\ref{unilateral coherence}), the left-hand side of Eq.\;(\ref{pati coherence}) can be rewritten as\\ 
\begin{eqnarray}
\small C_{rel}(\mathbb{X}\vert\rho_{A})+C_{rel}(\mathbb{Z}\vert\rho_{A})+2I(A:B)-I(X:B)-I(Z:B) &\geqslant & q_{MU}-S(\rho_{A})+2I(A:B)-I(X:B)-I(Z:B) \nonumber \\
&\geqslant & q_{MU}-S(\rho_{A})+2D_{A}(\rho_{AB}) \nonumber \\
&=& q_{MU}-S(A\vert B)-I(A:B)+2D_{A}(\rho_{AB})  \\
&=& q_{MU}-S(A\vert B)+D_{A}(\rho_{AB})-J_{A}(\rho_{AB}) \nonumber \\
& =& q_{MU} - S(A \vert B)+\max\lbrace 0,D_{A}(\rho_{AB})-J_{A}(\rho_{AB}) \rbrace, \nonumber
\end{eqnarray}
where the first inequality follows from the Eq.\;(\ref{unrc1}) and the second inequality follows from the definition of quantum discord.
 \end{proof}
One can also improve the lower bound of the uncertainty relation for the relative entropy of unilateral coherence by adding an additional term depending on the mutual information and the Holevo quantity.
\begin{theorem}
For two incompatible observables $X$ and $Z$ with bases $\mathbb{X}$ and $\mathbb{Z}$ , respectively, and any state $\rho_{AB}$, the following uncertainty relation for the relative entropy of unilateral coherence holds, 
\begin{equation}\label{adabi coherence}
C_{rel}^{B|A}(\mathbb{X}\vert\rho_{AB})+C_{rel}^{B|A}(\mathbb{Z}\vert\rho_{AB})\geqslant q_{MU} -S(A\vert B)+\max\{0 , \delta\}.
\end{equation}
\end{theorem} 
\begin{proof}
\begin{eqnarray}
C_{rel}(\mathbb{X}\vert\rho_{A})+C_{rel}(\mathbb{Z}\vert\rho_{A})+2I(A:B)-I(X:B)-I(Z:B)& \geqslant & q_{MU}-S(\rho_{A})+2I(A:B)-I(X:B)-I(Z:B) \nonumber \\
&=& q_{MU}-S(A\vert B)+I(A:B)-I(X:B)-I(Z:B) \nonumber \\
&=& q_{MU} -S(A\vert B)+\max\{0 , \delta\}.
\end{eqnarray}
\end{proof} 
As can be seen, the lower bound in Eq.\;(\ref{adabi coherence}) is tighter than that in Eq.\;(\ref{pati coherence}). It is also basis dependent. \\

To illustrate these results, let us consider a special class of two-qubit $X$ states 
\begin{equation}\label{Xstate}
\rho_{AB}= p\vert\Psi^{+}\rangle\langle\Psi^{+}\vert+(1-p)\vert11\rangle\langle11\vert,
\end{equation}
where $|\Psi^{+}\rangle=\frac{1}{\sqrt{2}}(|01\rangle+|10\rangle)$ is a maximally entangled state and $0\leqslant p\leqslant 1$. Two complementary observables measured on the part $A$ of this state are assumed to be the Pauli matrices, $X=\sigma_{1}$ and $Z=\sigma_{3}$. In Fig. \ref{fig1}, different lower bounds of the uncertainty relation for the relative entropy of unilateral coherence for this state are plotted versus the parameter $p$. As can be seen, the lower bound in Eq.\;(\ref{adabi coherence}) is tighter than those in Eqs.\;(\ref{berta coherence}) and (\ref{pati coherence}). \\
%%%%%%%%%%%%%%%%%%%%%%%%%%%%%%%%%%%%
%%%%%%%%%%%%%%%%%%%%%%%%%%%%%%%%%%%%
%%%%%%%%%%%%%%%%%%%%%%%%%%%%%%%%%%%%
\begin{figure}[ht] 
\centering
\includegraphics[width=8cm]{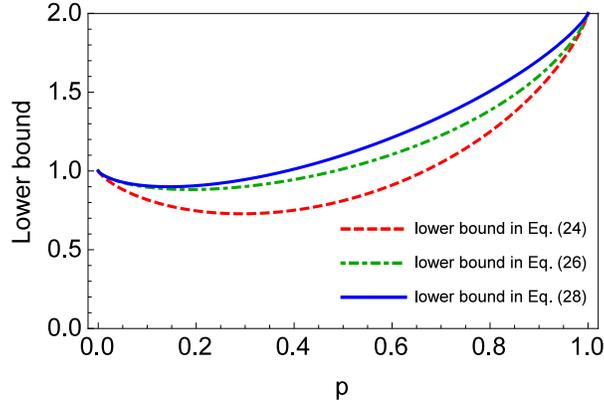}
\caption{(Color online) Different lower bounds of the uncertainty relation for the relative entropy of unilateral coherence for two complementary observables $X=\sigma_{1}$ and $Z=\sigma_{3}$ measured on the part $A$ of the state in  Eq.\;(\ref{Xstate}), versus the parameter $p$.}\label{fig1}
\end{figure}
%%%%%%%%%%%%%%%%%%%%%%%%%%%%%%%%%%%%
%%%%%%%%%%%%%%%%%%%%%%%%%%%%%%%%%%%%

Now, let us try to find out if there exist nontrivial upper bounds on the sum of the relative entropies of unilateral coherences of two general measurement bases. From Eq.\;(\ref{puritycoherence}) one can obtain a nontrivial upper bound on the sum of the relative entropies of unilateral coherences which is 
\begin{equation}\label{upperpurity}
2 \mathcal{P}_{rel}^{B|A}(\rho_{AB}) \geqslant C_{rel}^{B|A}(\mathbb{X}\vert\rho_{AB})+C_{rel}^{B|A}(\mathbb{Z}\vert\rho_{AB}).
\end{equation}
It is worth noting that if $X$ and $Z$ are complementary such that $q_{MU}=\log_2 d$, one has
\begin{eqnarray}
2 \mathcal{P}_{rel}^{B|A}(\rho_{AB})&\geqslant& C_{rel}^{B|A}(\mathbb{X}\vert\rho_{AB})+C_{rel}^{B|A}(\mathbb{Z}\vert\rho_{AB})\nonumber \\
&\geqslant& \mathcal{P}_{rel}^{B|A}(\rho_{AB}).
\end{eqnarray}
Also, when the subsystem $A$ is maximally mixed, the unilateral purity is equal to the mutual information, thus the above equation can be rewritten as
\begin{eqnarray}\nonumber
2 I(A:B)\geqslant C_{rel}^{B|A}(\mathbb{X}\vert\rho_{AB})+C_{rel}^{B|A}(\mathbb{Z}\vert\rho_{AB})\geqslant I(A:B).\\
\end{eqnarray}
It should be mentioned that the upper bound in Eq.\;(\ref{upperpurity}) can be improved by the Holevo quantity. To achieve this aim, one can use Eqs.\;(\ref{unilateral coherence}) and (\ref{puritycoherence}) as follows
\begin{eqnarray}
C_{rel}^{B|A}(\mathbb{Y}\vert\rho_{AB})&=&C_{rel}(\mathbb{Y}\vert\rho_{A})+I(A:B)-I(Y:B) \nonumber \\
&\leqslant &  \mathcal{P}_{rel}(\rho_{A})+I(A:B)-I(Y:B)\nonumber \\
&=&\mathcal{P}_{rel}^{B|A}(\rho_{AB})-I(Y:B).
\end{eqnarray}
Thus,
\begin{eqnarray}\label{upperpurity2}
2 \mathcal{P}_{rel}^{B|A}(\rho_{AB})-I(X:B)-I(Z:B) 
\geqslant C_{rel}^{B|A}(\mathbb{X}\vert\rho_{AB})+C_{rel}^{B|A}(\mathbb{Z}\vert\rho_{AB}).
\end{eqnarray}
This inequality becomes equality if the subsystem $A$ is maximally mixed. As can be seen, the upper bound in Eq.\;(\ref{upperpurity2}) depends on the observables, whereas the upper bound in Eq.\;(\ref{upperpurity}) is independent of the observables.   
As an example, let us consider a Bell diagonal state which can be written as
\begin{equation}
\rho_{AB}=\frac{1}{4}\left( \mathbf{I}_{A}\otimes\mathbf{I}_{B}+\sum^{3}_{i,j=1}w_{ij}\sigma_{i}\otimes\sigma_{j}\right) ,
\end{equation}
where $\sigma_{i} (i = 1,2,3)$ is the Pauli matrix. The correlation matrix, $W = \lbrace w_{ij} \rbrace$, can always be diagonalized by a suitable local unitary transformation, therefore one has
\begin{equation}\label{bell1}
\rho_{AB}=\frac{1}{4}\left( \mathbf{I}_{A}\otimes\mathbf{I}_{B}+\sum^{3}_{i=1}t_{i}\sigma_{i}\otimes\sigma_{i}\right) .
\end{equation}
Furthermore, when $t_{1}=1-2p$ and $t_{2}=t_{3}=-p$, with $0\leqslant p\leqslant 1$, the state in Eq. (\ref{bell1}) can be rewritten as
\begin{equation}\label{bellstate} 
\rho_{AB}=p \vert\Psi^{-}\rangle\langle\Psi^{-}\vert + \frac{1-p}{2}\left( \vert\Psi^{+}\rangle\langle\Psi^{+}\vert+\vert \Phi^{+} \rangle\langle \Phi^{+} \vert \right) ,\\
\end{equation}
%%%%%%%%%%%%%%%%%%%%%%%%%%%%%%%%%%%%
%%%%%%%%%%%%%%%%%%%%%%%%%%%%%%%%%%%%
%%%%%%%%%%%%%%%%%%%%%%%%%%%%%%%%%%%%
\begin{figure}[ht] 
\centering
\includegraphics[width=8cm]{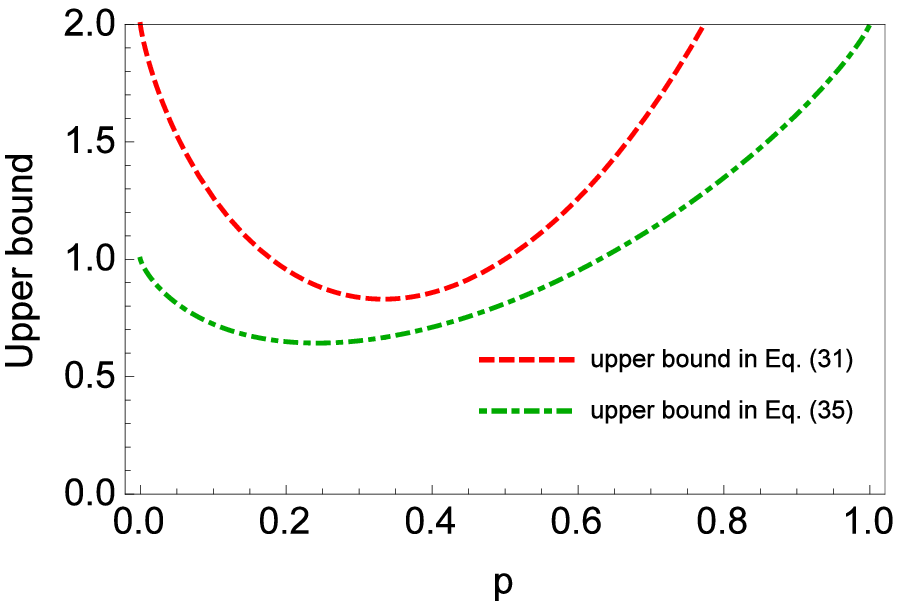}
\caption{(Color online) Different upper bounds of the uncertainty relation for the relative entropy of unilateral coherence for two complementary observables $X=\sigma_{1}$ and $Z=\sigma_{3}$ measured on the part $A$ of the state in  Eq.\;(\ref{bellstate}), versus the parameter $p$.}\label{fig2}
\end{figure}
%%%%%%%%%%%%%%%%%%%%%%%%%%%%%%%%%%%%
%%%%%%%%%%%%%%%%%%%%%%%%%%%%%%%%%%%%
%%%%%%%%%%%%%%%%%%%%%%%%%%%%%%%%%%%%
%%%%%%%%%%%%%%%%%%%%%%%%%%%%%%%%%%%%
%%%%%%%%%%%%%%%%%%%%%%%%%%%%%%%%%%%%
\begin{figure}[ht] 
\centering
\includegraphics[width=8cm]{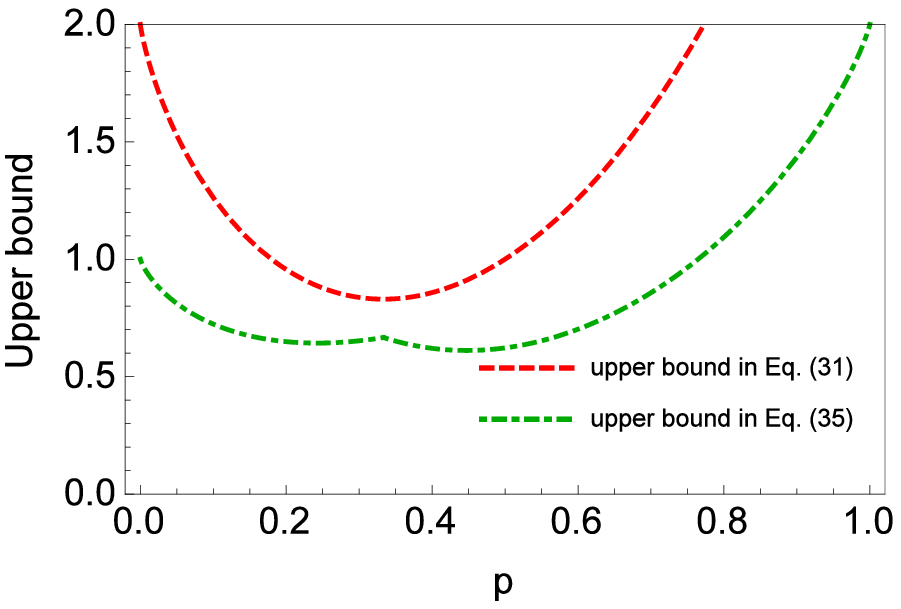}
\caption{(Color online) Different upper bounds of the uncertainty relation for the relative entropy of unilateral coherence for two complementary observables $X=\sigma_{1}$ and $Z=\sigma_{2}$ measured on the part $A$ of the state in  Eq.\;(\ref{bellstate}), versus the parameter $p$.}\label{fig3}
\end{figure}
%%%%%%%%%%%%%%%%%%%%%%%%%%%%%%%%%%%%
%%%%%%%%%%%%%%%%%%%%%%%%%%%%%%%%%%%%
where $\vert \Psi^{\mp} \rangle = \frac{1}{\sqrt{2}}(|01\rangle\mp|10\rangle)$, $\vert \Phi^{+} \rangle = \frac{1}{\sqrt{2}}(|00\rangle+|11\rangle)$ are the Bell states. Regarding this state, it turns out that there is the same upper bound for two sets of
complementary observables $\lbrace X=\sigma_{1} , Z=\sigma_{3} \rbrace$ and $\lbrace X=\sigma_{1} , Z=\sigma_{2} \rbrace$ which is 
\begin{equation}
2 \mathcal{P}_{rel}^{B|A}(\rho_{AB})=2\left( 2+p\log_{2}p+(1-p)\log_{2}(\frac{1-p}{2})\right) .
\end{equation}
However, the upper bounds in Eq.\;(\ref{upperpurity2}) for the above-mentioned two sets are 
\begin{eqnarray}
2 \mathcal{P}_{rel}^{B|A}(\rho_{AB})-I(X:B)-I(Z:B)&=&2\left( 2+p\log_{2}p+(1-p)\log_{2}(\frac{1-p}{2})\right)  \nonumber \\
&-&max\left\lbrace  1-h(p),1-h( \frac{1+p}{2})\right\rbrace  \nonumber \\
&-&min\left\lbrace  1-h(p),1-h(\frac{1+p}{2})\right\rbrace,
\end{eqnarray}
and
\begin{eqnarray}
2 \mathcal{P}_{rel}^{B|A}(\rho_{AB})-I(X:B)-I(Z:B)&=&2\left( 2+p\log_{2}p+(1-p)\log_{2}(\frac{1-p}{2})\right)  \nonumber \\
&-&max\left\lbrace  1-h(p),1-h(\frac{1+p}{2})\right\rbrace  -1+h(\frac{1+p}{2}),
\end{eqnarray}
respectively, where $h(x)=-x\log_{2}x-(1-x)\log_{2}(1-x)$ is the binary entropy.
In Figs. \ref{fig2} and \ref{fig3}, the upper bounds on the sum of the relative entropies of unilateral coherences for these observables are plotted in terms of the parameter $p$.

\section{quantum coherence and entropic uncertainty relations in the presence of a quantum memory}\label{Sec4}
In this section it is shown that regarding the uncertainty relations for the relative entropy of unilateral coherence, one can obtain the known EURs in the presence of a quantum memory. It turns out that this approach is very simple compared to the other methods.\\
Assume a bipartite system $AB$ in the state $\rho_{AB}$, if the measurement of the observable $Y \in \lbrace X, Z \rbrace$ is done on the part $A$, the post measurement state can be written as 
%%%%%%%%%%%%%%%%%%%%%%%%
%%%%%%%%%%%%%%%%%%%%%%%%%%
\begin{equation}
\rho_{YB}=\sum_{y}p_y \vert y\rangle\langle y \vert_{A}\otimes \rho_{B|y},
\end{equation}
where  $\vert y \rangle \in \lbrace \vert x \rangle,\vert z \rangle \rbrace$. $\rho_{YB}$ is the incoherent-quantum state with respect to $\lbrace\vert y \rangle\rbrace$ basis. Using the relative entropy of unilateral coherence introduced in Eq.\;(\ref{reou2}), one can obtain
\begin{equation}\label{tor1}
H(Y|B)=C_{rel}^{B|A}(\mathbb{Y}\vert\rho_{AB})+S(A|B).
\end{equation}
Considering Eq.\;(\ref{tor1}), for two incompatible quantum measurements (corresponding to two incompatible observables $X$ and $Z$) on the part $A$, one obtains
\begin{equation}\label{new}
H(X|B)+ H(Z|B)=C_{rel}^{B|A}(\mathbb{X}\vert\rho_{AB})+C_{rel}^{B|A}   (\mathbb{Z}\vert\rho_{AB})+2 S(A|B).
\end{equation}
Regarding Eq.\;(\ref{new}), the uncertainty relations for the relative entropy of unilateral coherence can be converted into the EURs in the presence of a quantum memory. This equation, along with Eq.\;(\ref{berta coherence}), leads us to
\begin{equation}\label{Berta22}
H(X \vert B)+H(Z \vert B) \geqslant q_{MU} +S(A \vert B),
\end{equation}
which is the Berta EUR in the presence of a quantum memory. The above derivation of Eq.\;(\ref{Berta22}) is much simpler than the one given in \cite{Berta}, which uses smooth entropies.\\
Similarly, one can obtain the Adabi and Pati EURs in the presence of a quantum memory using the relations obtained in Sec. \ref{Sec3}.
Considering Eqs.\;(\ref{pati coherence}) and (\ref{adabi coherence}), one can obtain the Pati (Eq.\;(\ref{pati})) and the Adabi (Eq.\;(\ref{new1})) EURs in the presence of a quantum memory, respectively. \\
Based on what has been mentioned so far, one comes to the result that there is a simple way  to convert an entropy-based uncertainty relation of unilateral coherence to the EUR in the presence of a quantum memory. According to our knowledge so far, it is actually the simplest method that has ever been provided to obtain the similar EURs.\\

In the bipartite system $AB$, entanglement between the part $A$ and the quantum memory $B$, $S(A \vert B)<0$, leads to a decrement in the lower bound of the EUR and an increment in that of uncertainty relation for the relative entropy unilateral coherence.  \\

Now, as an example, let us consider a two-qubit Werner state,
\begin{equation}\label{fig4} 
\rho_{AB}=p \vert \Psi^{+} \rangle\langle \Psi^{+} \vert + \frac{1-p}{4}\mathbf{I}_{A}\otimes\mathbf{I}_{B},
\end{equation}
 where $0\leqslant p\leqslant 1$. In Fig. \ref{fig4}, the lower bounds of the uncertainty relations for the relative entropy of unilateral coherence and the EUR in the presence of a quantum memory for this state are plotted versus the parameter $p$. As can be seen, when lower bound of coherence is larger than that of the EUR, the conditional entropy, $S(A \vert B)$, is negative meaning that $A$ and $B$ are entangled.

Finding nontrivial upper bounds on the sum of the entropies for different
observables in the presence of a quantum memory is largely an open problem. These bounds are called certainty relations in the presence of a quantum memory \cite{Coles1}.    The upper bound on the sum of the relative entropies of unilateral coherences, Eq.\;(\ref{upperpurity2}), can be used to derive a certainty relation in the presence of a quantum memory which is
\begin{equation}
2\log_{2} d-I(X:B)-I(Z:B) \geqslant H(X \vert B)+H(Z \vert B).
\end{equation}\\

It is interesting to note that the above inequality becomes an equality if the subsystem $A$ is maximally mixed. Thus, this upper bound is perfectly tight for the class of states with maximally mixed subsystem $A$ including Werner states, Bell diagonal states, and isotropic states.\\

%%%%%%%%%%%%%%%%%%%%%%%%%%%
\begin{figure}[ht] 
\centering
\includegraphics[width=8cm]{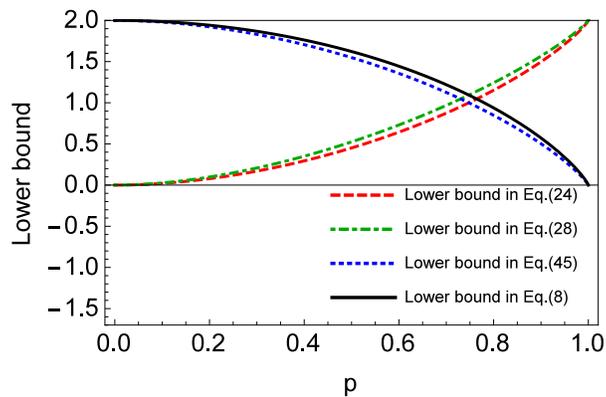}
\caption{(Color online) Lower bounds of uncertainty relation for the relative entropy of unilateral coherence and the EUR in the presence of a quantum memory for two complementary observables $X=\sigma_{1}$ and $Z=\sigma_{3}$, versus the parameter $p$ appeared in the expression for a correlated bipartite state which is assumed to be the Werner state.}\label{fig4}
\end{figure}

%%%%%%%%%%%%%%%%%%%%%%%%%%%%%

\section{Conclusion}\label{conclusion}
 In this work, the uncertainty relations for quantum coherence were obtained. To derive these relations, the properties of the relative entropy of quantum coherence have only been used. The most important property is that the relative entropy is nonincreasing under quantum operations. Regarding the  uncertainty relations for unilateral coherence, we could derive the uncertainty relations  in the presence of a quantum memory. We also provided two upper bounds on the sum of unilateral coherences defined based on the relative entropy in two incompatible reference bases and compared them. Using the upper bound on the sum of unilateral coherences, certainty relation in the presence of quantum memory was obtained.
%=============================================================%
%=============================================================%
%=======================  References =========================%
%=============================================================%
%=============================================================%



%% References


%\bibliography{sample.bib}



%%%%%%%%%%%%%%%%%%%%%%%%%%%%%%%%%%%%%%%%%%%%%%%%%%%%%%%%%%%%%%%%%%%%%
%%%%%%%%%%%%%%%%%%%%%%%%%%%%%%%%%%%%%%%%%%%%%%%%%%%%%%%%%%%%%%%%%%%%

\end{document}